\documentclass[a4paper,12pt]{article}
\usepackage[utf8]{inputenc}
\usepackage{graphicx}
\usepackage{amsfonts}
\usepackage{amsmath}
\usepackage{amsthm}
\usepackage{color}
\usepackage{hyperref}
\usepackage{geometry}
\usepackage[shortlabels]{enumitem}
\def\shuffle{\sqcup\mathchoice{\mkern-7mu}{\mkern-7mu}{\mkern-3.2mu}{\mkern-3.8mu}\sqcup}

\usepackage[ruled, linesnumbered]{algorithm2e}
\SetKw{Continue}{continue}
\SetKw{And}{and}
\makeatletter
\newcommand\fs@spaceruled{\def\@fs@cfont{\bfseries}\let\@fs@capt\floatc@ruled
  \def\@fs@pre{\vspace{5\baselineskip}\hrule height.8pt depth0pt \kern2pt}%
  \def\@fs@post{\kern2pt\hrule\relax}%
  \def\@fs@mid{\kern2pt\hrule\kern2pt}%
  \let\@fs@iftopcapt\iftrue}
\makeatother

\usepackage{caption}
\usepackage{subcaption}
\usepackage{mathrsfs}  
\usepackage{amssymb}
\usepackage[title]{appendix}
\newtheorem{theorem}{Theorem}[section]
\newtheorem{proposition}[theorem]{Proposition}
\newtheorem{corollary}[theorem]{Corollary}
\newtheorem{lemma}[theorem]{Lemma}

\newtheorem{definition}[theorem]{Definition}
\theoremstyle{remark}

\newtheorem{example}[theorem]{Example}
\usepackage{authblk}

\def\word#1{{\color{blue}\mathbf{#1}}}

\date{\today}
\title{Optimal execution with rough path signatures\footnote{Opinions expressed in this paper are those of the authors, and do not necessarily reflect the view of JP Morgan.}}

\usepackage{authblk}

\author[1]{Jasdeep Kalsi}
\author[1, 2]{Terry Lyons}
\author[1,2,3]{Imanol Perez Arribas}
\affil[1]{{\em Mathematical Institute, University of Oxford}}
\affil[2]{{\em The Alan Turing Institute, London}}
\affil[3]{{\em J.P. Morgan, London}}

\usepackage{mathrsfs}
\usepackage{ifthen}
\usepackage{verbatim}
\usepackage{fancyhdr}
\usepackage{enumitem}
\usepackage{amsmath,amssymb,amsthm,graphicx}
\usepackage{latexsym}
\usepackage{ifpdf}
\usepackage{appendix}
\usepackage{color}

\theoremstyle{definition}

\def\r{\right}

\usepackage{setspace}

\usepackage{subcaption}

\usepackage[super]{nth}

\begin{document}
{\setstretch{1}
\maketitle
}

\begin{abstract}
We present a method for obtaining approximate solutions to the problem of optimal execution, based on a signature method. The framework is general, only requiring that the price process is a geometric rough path and the price impact function is a continuous function of the trading speed. Following an approximation of the optimisation problem, we are able to calculate an optimal solution for the trading speed in the space of linear functions on a truncation of the signature of the price process. We provide strong numerical evidence illustrating the accuracy and flexibility of the approach. Our numerical investigation both examines cases where exact solutions are known, demonstrating that the method accurately approximates these solutions, and models where exact solutions are not known. In the latter case, we obtain favourable comparisons with standard execution strategies.

\end{abstract}


\section{Introduction}

\subsection{Overview}

The problem of optimal execution has attracted much interest following the original work on the problem by Bertsimas and Lo in \cite{bertsimas} and Almgren and Chriss in \cite{almgren}. The aim is to model how one should send orders to the market in order to transition from holding one portfolio to another. Typically the case where an investor simply wishes to acquire/liquidate shares in a single asset is considered. There are two competing factors to the optimisation. Firstly, the investor has pressure to trade quickly. Trading more at later times would mean accepting more risk, as the future prices are uncertain. On the other hand, trading evenly across time also has its benefits due to the nature of market impact. The investor should consider how much liquidity there currently is at desirable prices -- placing a large order now could result in ``walking the book" and accepting unfavourable prices for a large portion of their trade. 

The key features in any optimal execution model are the dynamics of the price process at which the trader can execute her trades, $P_t$, and some definition of the notion of a good strategy for the trader. The process $P_t$ is a function of the history of the trading speed until time $t$, together with some additional driving processes. Typically, we have that $P_t$ is given by the sum of an underlying price process added to a price impact function. The price impact function depends on the history of the investor's trading speed, and it determines how much the price at which the trader can execute has changed as a consequence of that. Classical choices of price impact functions include \emph{temporary} versions, which depend only on the speed at which the trader wishes to trade at that time, \emph{permanent} versions, which depend in the accumulation of orders placed until time $t$, and \emph{transient} versions, where the effects of past trading speeds decay with time. Good strategies are usually defined in terms of some cost functional, which takes into account both the expected revenue for the investor when employing a strategy, and some measure of the risk associated with that strategy.

\subsection{Paper Outline}

The aim of this paper is to show how the signature method can be used to obtain approximate solutions to the problem of optimal execution. Our setting is very general, with price processes assumed to be geometric rough paths only and the price impact function allowed to depend on the entire history of the trading strategy, with only a mild continuity condition assumed. The flexibility of the framework is demonstrated in part by the broad range of existing models in the literature which fall within it. An instance of this is the classical optimal execution problem presented in \cite[Section 6.5]{carteabook}, in which the underlying price process is assumed to be a Brownian motion, with an $L^2$ penalty imposed based on the risk of holding inventory. More recent examples include the work of Lehalle and Neumann in \cite{eyal} and Cartea and Jaimungal in \cite{cartea2}. In \cite{eyal}, the authors prove results on the existence and uniqueness of an optimal trading strategy in the setting where trading signals are incorporated into the price dynamics. Similarly in \cite{cartea2}, the authors consider the role of microstructure in the problem by including order flow as a contributing factor permanently affecting the price. Our approach can also be adapted to handle models consisting of multiple correlated assets which are effected by trades in each other. Such a setting is presented in the article by Mastromatteo, Benzaquen, Eisler and Bouchaud, \cite{crossasset2}.

We begin the paper with a brief overview of rough paths in Section \ref{sec:rough paths}. Here, we define geometric rough paths and their signatures, and introduce the underlying algebraic structures required to perform calculations on the signatures. Following this, we introduce our framework in Section \ref{sec:framework}. This consists of specifying our assumptions on the price process and market impact in our model, defining the space of trading speeds in which we will look for strategies, and introducing the optimal control problem. Section \ref{sec:optimal execution} is dedicated to calculating approximate solutions to the control problem. We first reformulate the problem in terms of the signature, and then we approximate the optimal trading speed by a finite-dimensional, computationally tractable minimisation problem. In Section \ref{sec:extensions}, we provide examples of interesting extensions of the approach as it was presented in Sections \ref{sec:framework} and \ref{sec:optimal execution}, such as the multiple asset problem which appears in \cite{crossasset2}, and more exotic models where additional multi-dimensional noise is assumed to provide exogenous information about the price dynamics. Finally in Section \ref{sec:numerical experiments} and Section \ref{sec:market data}, we provide numerical evidence that the model performs well. Good approximations to the optimal strategies in the settings \cite[Section 6.5]{carteabook}, \cite{eyal} and \cite{cartea2} are obtained, and we also investigate the problem in the case where the underlying price process is a fractional Brownian motion. Moreover, we demonstrate in Section \ref{sec:market data} that our methodology can also be used on real market data.

\section{Rough paths preliminaries}\label{sec:rough paths}

Rough paths and signatures will play a key role in this paper. In this section we will introduce all the aspects of rough paths theory that will be used in the article. For a more detailed introduction to the theory of rough paths, the authors refer the reader to \cite{lyonsbook, frizvictoir}.

\subsection{Tensor algebra}\label{sec:tensor algebra}

A rough path is a path that takes value on a certain graded space, called the tensor algebra. This subsection will introduce these algebras, as well as another crucial space -- the dual space of the tensor algebra.

\begin{definition}[Extended tensor algebra]
Let $d\geq 1$. We denote by $T((\mathbb R^d))$ the extended tensor algebra over $\mathbb R^d$, which is defined by
$$T((\mathbb R^d)) := \{\mathbf a = (a_0, a_1, \ldots, a_n, \ldots)\;|\;a_n\in (\mathbb R^d)^{\otimes n}\}$$
where $\otimes$ denotes the tensor product. Given $\mathbf{a}=(a_i)_{i=0}^\infty, \mathbf{b}=(b_i)_{i=0}^{\infty}\in T((\mathbb R^d))$, define the sum $+$ and product $\otimes$ by \begin{align*}
&\mathbf a + \mathbf b := (a_i + b_i)_{i=0}^\infty,\\
&\mathbf a \otimes \mathbf b := \left (\sum_{k=0}^i a_k\otimes b_{i-k}\right )_{i=0}^\infty.
\end{align*}
We also define the action on $\mathbb R$ given by $\lambda \mathbf a := (\lambda a_i)_{i=0}^\infty$ for all $\lambda\in \mathbb R$, $\mathbf a \in T((\mathbb R^d))$.
\end{definition}

Similarly, we can define the tensor algebra and truncated tensor algebra as the space of all finite sequences and all sequences of a given length, respectively.

\begin{definition}
The tensor algebra over $\mathbb R^d$, denoted by $T(\mathbb R^d) \subset T((\mathbb R^d))$, is given by
$$T(\mathbb R^d) := \{\mathbf a = (a_n)_{n=0}^\infty\;|\;a_n\in (\mathbb R^d)^{\otimes n}\mbox{ and }\exists N\in \mathbb N \mbox{ such that }a_n=0\,\forall n\geq N\}.$$
Similarly, the truncated tensor algebra of order $n\in \mathbb N$ over $\mathbb R^d$ is defined by
$$T^{(N)}(\mathbb R^d) := \{\mathbf a = (a_n)_{n=0}^\infty\;|\;a_n\in (\mathbb R^d)^{\otimes n}\mbox{ and }a_n=0\,\forall n\geq N\}.$$
\end{definition}

Let $\{e_1, \ldots, e_d\}\subset \mathbb R^d$ be a basis for $\mathbb R^d$. This induces a dual basis $\{e_1^\ast, \ldots, e_d^\ast\}\subset (\mathbb R^d)^\ast$ for $(\mathbb R^d)^\ast$, where $(\mathbb R^d)^\ast$ denotes the dual space of $\mathbb R^d$ -- i.e. the space of all continuous linear functions $\mathbb R^d \to \mathbb R$. We may define a basis for $(\mathbb R^d)^{\otimes n}$ by:
$$\{e_{i_1} \otimes \ldots \otimes e_{i_n} \;\mid\;i_j\in \{1, \ldots, d\} \mbox{ for }j=1, \ldots, n\}.$$
Similarly, a basis of $((\mathbb R^d)^\ast)^{\otimes n}$ is defined by
$$\{e_{i_1}^\ast \otimes \ldots \otimes e_{i_n}^\ast \;\mid\;i_j\in \{1, \ldots, d\} \mbox{ for }j=1, \ldots, n\}.$$
This induces, in a natural way, a basis for $T((\mathbb R^d))$ and $T((\mathbb R^d)^\ast)$.

It is often convenient to think of $T((\mathbb R^d)^\ast)$ as a space of \textit{words}. Define the alphabet $\mathcal A_d := \{\word 1, \ldots, \word d\}$. Then, the basic element $e_{i_1}^\ast \otimes \ldots \otimes e_{i_n}^\ast$ can be identified with the word $\word{i_1}\ldots\word{i_n}$. Let $\mathcal W(\mathcal A_d)$ denote the space of all words (and their sums) with letters in the dictionary $\mathcal A_d$, i.e. the free $\mathbb R$-vector space generated by $\mathcal A_d$. Then, we have $T((\mathbb R^d)^\ast) \cong \mathcal W(\mathcal A_d)$. The empty word will be denoted by $\word \varnothing \in \mathcal W(\mathcal A_d)$.

\begin{example}Consider the following examples for $\mathbb R^2$.
\begin{enumerate}
\item Let $\mathbf a = 3 - e_2\otimes e_1\in T((\mathbb R^2))$. Then, $\langle \word \varnothing, \mathbf a \rangle = 3$.
\item Let $\mathbf a = 1 - 2e_1 + e_2\in T((\mathbb R^2))$, and set $\word w = \word{\varnothing} + \word{1}$. Then, $\langle \word w, \mathbf a \rangle = 1 - 2 = -1$.
\item Let $\mathbf a = e_1 \otimes e_2 - e_2 \otimes e_1\in T((\mathbb R^2))$, and set $\word w = \word{21} + \word{111}$. Then, $\langle \word w, \mathbf a \rangle = -1 + 0 = -1$.
\item Let $\mathbf a = 1 + e_1^{\otimes 3}$ and $\word w = 2 \cdot \word{111}$. Then, $\langle \word w , \mathbf a \rangle = 2 \cdot 1 = 2$.
\end{enumerate}
\end{example}
The space of words possesses two natural algebraic operations -- the sum and the concatenation. Let $\word w=\word{i_1}\ldots\word{i_n}, \word v=\word j_1\ldots \word{j_m}\in \mathcal W(\mathcal A_d)$ be two words. Their sum is the formal sum $\word w + \word v\in \mathcal W(\mathcal A_d)$. Their concatenation, on the other hand, is defined by $$\word{wv} := \word{i_1}\ldots\word{i_nj_1}\ldots\word{j_m}\in \mathcal W(\mathcal A_d).$$ These two operations induce analogous operations on $T((\mathbb R^d)^\ast)$, and with some abuse of notation we will even use concatenation on $\mathcal W(\mathcal A_d)$ and $T((\mathbb R^d)^\ast)$ interchangeably -- i.e. we will sometimes write $\ell \word w \in T((\mathbb R^d)^\ast)$ for $\ell \in T((\mathbb R^d)^\ast)$ and word $\word w\in \mathcal W(\mathcal A_d)$, by which we mean that we take the concatenation of the element in $\mathcal W(\mathcal A_d)$ associated to $\ell$ and the word $\word w$.

\begin{example} Take the alphabet $\mathcal A_4 = \{\word 1, \word 2, \word 3, \word 4\}$.

\begin{enumerate}
\item Set $\word w=\word{212}, \word{v}=\word{31}$. We have $\word {wv}=\word{21231}\in \mathcal W(\mathcal A_4)$.
\item We have $(\word{143} + \word{23})\word{1} = \word{1431} + \word{231}\in \mathcal W(\mathcal A_4)$.
\end{enumerate}
\end{example}
There is a third operation on words that will be useful in this paper: the \textit{shuffle product} $\shuffle\phantom{ }$. Intuitively, the shuffle product accounts for all the possible ways of riffle shuffling two decks of cards. The precise definition is given below.

\begin{definition}[Shuffle product]
The shuffle product $\shuffle\phantom{ }:\mathcal W(\mathcal A_d)\times \mathcal W(\mathcal A_d)\to \mathcal W(\mathcal A_d)$ is defined inductively by $$\word{ua}\shuffle \word{vb}=(\word u\shuffle \word{vb})\word{a}+(\word{ua}\shuffle \word v)\word b,$$ $$\word w \shuffle \word \varnothing = \word \varnothing \shuffle \word w = \word w$$ for all words $\word{u},\word{v}$ and letters $\word a, \word b\in \mathcal A_d$, which is then extended by bilinearity to $\mathcal W(\mathcal A_d)$. With some abuse of notation, the shuffle product on $T((\mathbb R^d)^\ast)$ induced by the shuffle product on words will also be denoted by $\shuffle\phantom{ }$.
\end{definition}
It follows from the definition of the shuffle product that $\phantom{ }\shuffle\phantom{ }$ is commutative, i.e. $f \shuffle g = g \shuffle f$ for all $f,g\in T((\mathbb R^d)^\ast)$.

\begin{example}
We have:

\begin{enumerate}
\item $\word{12} \shuffle \word 3 = \word{123}+\word{132}+\word{312}$.
\item $\word{12}\shuffle \word{23} = 2 \cdot \word{1224} + \word{1242} + \word{2124} + \word{2142} + \word{2412}$.
\end{enumerate}
\end{example}

\begin{definition}\label{def:shuffle polynomial}
Let $Q\in \mathbb R[x]$ be a polynomial on one variable. Write $Q(x) = a_0 + a_1x + a_2x^2  + \ldots + a_nx^n$. Then, $Q$ induces the map $Q^{\shuffle\phantom{ }}:T((\mathbb R^d)^\ast)\to T((\mathbb R^d)^\ast)$ given by $$Q^{\shuffle\phantom{ }}(\ell) := a_0\word{\varnothing} + a_1\ell + a_2 \ell^{\shuffle 2} + \ldots + a_n \ell^{\shuffle n} \quad \forall \ell\in T((\mathbb R^d)^\ast),$$ where $\ell^{\shuffle k} :=\underbrace{\ell\shuffle \ell \shuffle \ldots \shuffle \ell}_{k}$ for $k\in \mathbb N$.
\end{definition}

\subsection{Rough paths}

We will now define a crucial object in this paper: the signature of a path.

\begin{definition}[Signature of a path]\label{def:sig}
Let $0\leq s < t \leq T$. For a piecewise smooth path $X:[0, T]\to \mathbb R^d$, we define the signature of $X$ over $[s, t]$ by $$\mathbb X_{s, t}^{<\infty} := (1, \mathbb X^1_{s,t}, \ldots, \mathbb X^n_{s,t}, \ldots)\in T((\mathbb R^d))$$ where $$\mathbb X^n_{s,t} := \int_{s<u_1<\ldots<u_n<t} dX_{u_1}\otimes \ldots \otimes dX_{u_n}\in (\mathbb R^d)^{\otimes n}.$$ Similarly, we define the truncated signature of order $N\in \mathbb N$ by $$\mathbb X_{s, t}^{\leq N} := (1, \mathbb X^1_{s,t}, \ldots, \mathbb X^N_{s,t})\in T^{(N)}(\mathbb R^d).$$
If we refer to the signature of $X$, without referencing the interval over which the signature is taken, we will implicitly refer to $\mathbb X_{0, T}^{<\infty}$.
\end{definition}

\begin{example}
Throughout this paper, we will constantly work with linear functions on the signature. Therefore, it will be useful to see a few examples that will be used in later sections.

Let $X=(X^1, X^2)\in C^\infty([0, T]; \mathbb R^2)$ be a two-dimensional smooth path. Recall that in Section \ref{sec:tensor algebra} we introduced the notation of \textit{words} as linear functions on the tensor algebra. We have:

\begin{enumerate}
\item $\langle \word 2, \mathbb X_{0,T}^{<\infty}\rangle=\int_0^T dX_t^2 =X^2_T - X^2_0$.
\item $\langle \word \varnothing, \mathbb X_{0,T}^{<\infty}\rangle = 1$.
\item $\langle \word{21}, \mathbb X_{0,T}^{<\infty}\rangle = \int_0^T \int_0^t dX_s^2 dX_t^1=\int_0^T (X_t^2-X_0^2)dX_t^1$.
\item Let $\ell\in T((\mathbb R^2)^\ast)$. Then, $\langle \ell\word 1, \mathbb X_{0,T}^{<\infty}\rangle = \int_0^T \langle \ell, \mathbb X_{0,t}^{<\infty}\rangle dX_t^1$.
\end{enumerate}
\end{example}

\begin{definition}[Geometric $p$-rough paths]
Let $T>0$ and $p\geq 1$. Denote by $\lfloor p \rfloor$ the integer part of $p$. Let $\Delta_T := \{(s, t)\in [0, T]\times [0, T]\; |\; s\leq t\}$. A function $\mathbb X:\Delta_T\to T^{(\lfloor p \rfloor)}(\mathbb R^d)$ is said to be a geometric $p$-rough path if it is the limit (under the $p$-variation distance, \cite[Definition 1.5]{lyonsbook}) of signatures of order $\lfloor p \rfloor$ of piecewise smooth paths. The space of all geometric $p$-rough paths will be denoted by $G\Omega_p([0, T]; \mathbb R^d)$.
\end{definition}
Each $\mathbb X=(1, \mathbb X^1, \ldots, \mathbb X^{\lfloor p \rfloor})\in G\Omega_p([0, T]; \mathbb R^d)$ can be uniquely extended to a $N$-geometric rough path for any $N\geq p$ (\cite[Theorem 3.7]{lyonsbook}). Analogously to the smooth case, the full extension $\mathbb X^{<\infty} = (1, \mathbb X^1, \ldots, \mathbb X^N, \ldots)$ will be defined as the \textit{signature} of $\mathbb X$.

Many stochastic processes that are used in the literature are almost surely geometric rough paths. For example, the signature of a semimartingale, defined using Stratonovich integration, is almost surely a geometric $p$-rough path for any $p\in (2, 3)$ \cite{semimartingalesgeometric}. The signature of a fractional Brownian motion for Hurst parameter $H\geq 1/4$, defined almost surely, is also a geometric $p$-rough path for $p> 1/H$ (\cite{fractionalbmgeometric}). We will now state some properties of signatures that will be useful in this article.

\begin{lemma}[Shuffle product property, \cite{lyonsbook}]
Let $\mathbb X\in G\Omega_p([0, T]; \mathbb R^d)$ be a geometric $p$-rough path. Let $\ell_1, \ell_2\in T((\mathbb R^d)^\ast)$ be two linear functionals. Then,
\begin{equation}\label{eq:shuffle product property}
\langle \ell_1, \mathbb X_{0, T}^{<\infty}\rangle\langle \ell_2, \mathbb X_{0, T}^{<\infty}\rangle = \langle \ell_1\shuffle \ell_2, \mathbb X_{0, T}^{<\infty}\rangle\quad \forall \ell_1,\ell_2\in T((\mathbb R^d)^\ast).
\end{equation}
\end{lemma}
The shuffle product will be extensively used throughout this paper. It guarantees that the product of two linear functions on the signature is another linear function on the signature, which is given explicitly in terms of the shuffle product.

The following lemma will also be useful in this paper. This result guarantees that the signature $\mathbb X_{0,T}^{<\infty}$ completely characterises $\mathbb X$ -- up to the so-called tree-like equivalences (see \cite[Definition 1.1]{horatio}).

\begin{lemma}[Uniqueness of signatures, \cite{horatio}]
Let $\mathbb X\in G\Omega_p([0,T];\mathbb R^d)$. The signature $\mathbb X_{0,T}^{<\infty}$ of $\mathbb X$ is unique up to tree-like equivalences (defined in \cite[Definition 1.1]{horatio}).
\end{lemma}

\begin{corollary}\label{cor:uniqueness}
Let $\mathbb X\in G\Omega_p([0,T];\mathbb R^d)$. If there exists a projection of $\mathbb X$ that is strictly monotone, then the signature $\mathbb X_{0,T}^{<\infty}$ determines $\mathbb X$ up to translations.
\end{corollary}

\section{Framework}\label{sec:framework}

\subsection{Notation}

Given a continuous path $Z\in C([0, T]; \mathbb R)$, denote its \textit{augmentation} by the continuous path $\widehat Z\in C([0, T]; \mathbb R_+\times \mathbb R)$ defined by $\widehat Z_t := (t, Z_t) \in \mathbb R_+ \times \mathbb R$.

Let $p \geq 1$. Define:
$$\widehat{\Omega}_T^p := \overline{\{ \mathbb{\widehat Z}\in G\Omega_p([0, T]; \mathbb R^2)\;\mid\;Z\in C^\infty([0, T]; \mathbb R)\mbox{ and }Z_0=1\}}^{d_{p-var}},$$
where the closure is taken under $d_{p-var}$, i.e. the $p$-variation distance (see \cite[Definition 1.5]{lyonsbook}). Given $\widehat{\mathbb Z}\in \widehat{\Omega}_T^p$, we will write by $Z\in C([0, T]; \mathbb R)$ the unaugmented coordinate process.

Intuitively, elements of $\widehat{\Omega}_T^p$ are signatures of paths of the form $(t, Z_t)$, with initial point $Z_0=1$. Because the first dimension of this augmented path (namely, time) is monotone increasing, and because we are only considering paths that start at 1, it follows by Corollary \ref{cor:uniqueness} that $\widehat{\mathbb Z}_{0,T}^{<\infty}$ completely characterises $\widehat{\mathbb Z}$ (and hence $Z$).

\subsection{The market}

The space $\widehat \Omega_T^p$ will be our space of market paths. We will equip it with a probability space $( \widehat{\Omega}_T^p, \mathcal B(\widehat{\Omega}_T^p), \mathbb P)$. Given a rough path $\widehat{\mathbb X}\in \widehat{\Omega}_T^p$, the unaugmented coordinate path $X:[0, T] \to \mathbb R$ will denote the \textit{unaffected midprice} of the asset. In other words, $X$ is the midprice process of the asset if the trader does not trade on the asset.

\begin{example}
Our market framework is very general in the sense that it includes most of the examples that have been considered in the literature. In particular, our framework includes:

\begin{enumerate}
\item \textbf{Semimartingales}. In the literature \cite{cartea1, cartea2, cartea3, eyal}, the midprice process is often modelled as a semimartingale. Semimartingales can be lifted to $p$-geometric rough paths for $p\in (2, 3)$ \cite{lyonsoriginal, semimartingalesgeometric, frizvictoir}, and therefore they fit into our framework: the market would be given by the probability space $(\widehat{\Omega}_T^p, \mathcal B(\widehat{\Omega}_T^p), \mathbb P)$ for $p\in (2,3)$ and $\mathbb P$ the law of the semimartingale.
\item \textbf{L\'evy processes}. More generally, certain L\'evy processes can also be lifted into $p$-geometric rough paths \cite{levy, levy2} and they are thus included in this framework.
\item \textbf{Fractional brownian motion}. Our framework also includes the setting where the midprice is modelled by a fractional brownian motion with Hurst parameter $H\geq 1/4$. Indeed, it was shown in \cite{fractionalbmgeometric} that fractional Brownian motions with Hurst parameter greater than 1/4 can be lifted to geometric rough paths.
\end{enumerate}
\end{example}

\subsection{Trading speeds}

In this section, we will introduce the notation of \textit{trading speeds}.

\begin{definition}[Trading speeds]
Define the metrizable space $\Lambda_T := \bigcup_{t\in [0,T]} \widehat{\Omega}_t^p$. We define the space of trading speeds by $\mathcal T:=C(\Lambda_T; \mathbb R)$. Given a trading speed $\theta\in \mathcal T$, the trader will trade a rate of $\theta(\widehat{\mathbb X}|_{[0,t]})$.
\end{definition}
Intuitively, the trader that is sitting at time $t\in [0, T]$ should decide how much to sell or buy by only considering what happened up to time $t $: she can only act based on the past, not the future. In other words, the trader's trading decision will be a (non-anticipative) function of the midprice process up to time $t$, i.e. $\widehat{\mathbb X}|_{[0,t]}\in \Lambda_T$. This intuition is incorporated into the definition of the trading speeds $\mathcal T$. A space similar to $\Lambda_T$ was considered in \cite{lambdaspace1, lambdaspace2, lambdaspace3, lambdaspace4, lambdaspace5, lambdaspace6}, and a similar definition of trading strategies was considered in \cite{lambdaspace6}.

In this paper, the following class of trading speeds will have a special relevance:

\begin{definition}[Signature trading speeds]
The space of signature trading speeds $\mathcal T_{sig}\subset \mathcal T$ is defined by $$\mathcal T_{sig} := \{ \theta \in \mathcal T \; \mid \; \exists \ell\in T((\mathbb R^2)^\ast) \mbox { such that }\theta(\widehat{\mathbb X}|_{[0, t]}) = \langle \ell, \widehat{\mathbb X}_{0, t}^{<\infty}\rangle\; \forall\, \widehat{\mathbb X}|_{[0, t]} \in \Lambda_T\}$$ where $\widehat{\mathbb X}_{0,t}^{<\infty}$ denotes the signature of $\widehat{\mathbb X}$ over the interval $[0, t]$.
\end{definition}
It turns out that the space of signature trading speeds $\mathcal T_{sig} \subset \mathcal T$ is very large -- in fact, we have the following density result, whose proof is in Appendix \ref{appendix:proofs}.

\begin{lemma}\label{lemma:density}
Let $\varepsilon > 0$. Then, there exists a compact set $\mathcal K\subset \widehat \Omega_T^p$ such that:

\begin{enumerate}
\item $\mathbb P[\mathcal K] > 1-\varepsilon$.
\item $\mathcal T_{sig}$, restricted to $\mathcal K$, is dense in $\mathcal T$.
\end{enumerate}
\end{lemma}
Therefore, trading speeds can be locally approximated arbitrarily well by signature trading speeds. Hence, if one wants to optimise a certain objective function over $\mathcal T$, it makes sense to optimise it over $\mathcal T_{sig}$ instead. This is precisely the approach that will be followed in this paper: we will look for an optimal trading speed in $\mathcal T_{sig}$, instead of $\mathcal T$.

\subsection{Market impact}\label{subsec:market impact}
When a trader buys or sells a traded asset, the mere act of trading will affect the asset's order book. If the volume she trades is small compared to the overall volume, this effect may be neglected. However, if the trader sends large trading orders the impact on the order book may negatively affect the price at which the order is executed (see \cite{marketimpact} and the references therein). In this section we will introduce the market impact model that will be used in this paper.

If the trader decides to follow a signature trading speed $\theta\in \mathcal T_{sig}$, the \textit{execution price} -- i.e. the price the trader has access to -- will be given by
\begin{equation}\label{eq:execution price}
P_t^\theta := X_t - \langle g^\theta, \widehat{\mathbb X}^{<\infty}_{0, t}\rangle,
\end{equation} where $g^\theta\in T((\mathbb R^2)^\ast)$ is a linear functional that depends on $\theta$ that models the market impact.

\begin{example}\label{example: impact}
The definition of the market impact, far from being restrictive, is very general and includes many examples that have been studied in the literature. Indeed, let $\ell \in T((\mathbb R^2)^\ast)$ and set the signature trading speed $\theta (\widehat{\mathbb X}|_{[0, t]}) := \langle \ell, \widehat{\mathbb X}_{0, t}\rangle $. Then, the following are examples of market impacts included in our framework:

\begin{enumerate}
\item \textbf{Temporary market impact}. Set $g^\ell := \lambda\ell$, with $\lambda>0$. Then, $\langle g^\ell, \widehat{\mathbb X}_{0,t}^{<\infty}\rangle = \lambda\theta(\widehat{\mathbb X}|_{[0, t]})$ is the linear temporary market impact studied in \cite{cartea1, cartea2, eyal}. We may also make the temporary market impact nonlinear by considering a polynomial $Q\in \mathbb R[x]$ and setting $g^\ell := Q^{\shuffle}(\ell)$. Then, $\langle g^\ell, \widehat{\mathbb X}_{0,t}^{<\infty}\rangle = Q(\theta(\widehat{\mathbb X}|_{[0, t]}))$.
\item \textbf{Permanent market impact.} In \cite{cartea1, cartea2, cartea3}, a permanent market impact given by $\int_0^t \theta_s ds$ is considered. Setting $g^\ell := \ell \word{1}$, we have $\langle g^\ell, \widehat{\mathbb X}_{0,t}^{<\infty}\rangle = \langle \ell \word{1}, \widehat{\mathbb X}_{0,t}^{<\infty}\rangle =\int_0^t \langle \ell, \widehat{\mathbb X}|_{[0, s]}\rangle ds = \int_0^t \theta(\widehat{\mathbb X}|_{[0, s]})ds$.
\item \textbf{Transient market impact}. In \cite{transient1, transient2, transient3} the authors considered a transient market impact that is given by $\int_0^t K(t-s)\theta_s ds$, where $K(x):= \exp(-\rho x)$ for $\rho>0$ constant. Then, we can find $g^\ell\in T((\mathbb R^2)^\ast)$ such that $$\int_0^t K(t-s)\theta_sds \approx \langle g^\ell, \widehat{\mathbb X}_{0,t}\rangle$$ to arbitrary accuracy.

\item More generally, market impacts modelled by functions of the form $G(\theta, X)$ can be well-approximated by linear functions on the signature, and they are therefore included in our framework.
\end{enumerate}

\end{example}

\subsection{Optimal execution problem}
Suppose the trader wishes to liquidate $q_0>0$ units of the asset by time $T$. If $q_0$ is large compared to the traded volume, the trading activity will affect the price of the asset (\cite{marketimpact}) negatively for the trader. Therefore, it may be more beneficial to spread the trading activity over the interval $[0, T]$ to avoid the undesired market impact. In this case, however, the trader will be exposed to market fluctuations that may affect her adversely. Hence, the task is to find a suitable trading speed to liquidate the inventory $q_0$ which accounts for this trade-off. We will now introduce the optimal execution problem that will be studied in this paper.

\begin{definition}\label{def:wealth inventory cost} The wealth corresponding to the trading speed $\theta\in \mathcal T$ is defined by

$$W_t^\theta := \int_0^t P_s^\theta\theta(\widehat{\mathbb X}|_{[0,s]}) ds.$$ On the other hand, the remaining inventory is defined by $$Q_t^\theta := q_0 - \int_0^t \theta(\widehat{\mathbb X}|_{[0, s]})ds$$ where $q_0>0$ is the initial inventory. We define the cost function $\mathcal C^\theta:\widehat \Omega_T \to \mathbb R$ by
\begin{equation}\label{eq:cost}
\mathcal C^\theta(\widehat{\mathbb X}) :=W_T^\theta - \phi \int_0^T                                                                                               (Q_t^\theta)^2 dt + Q_T^\theta (P_T^\theta - \alpha Q_T^\theta)\end{equation}
with $\alpha, \phi \geq 0$ constants.

\end{definition}
In this paper we will study the following optimal execution problem given by the optimisation problem

\begin{equation}\label{optimal execution problem}
\sup_{\theta\in \mathcal T} \mathbb E[\mathcal C^\theta (\widehat {\mathbb X})].
\end{equation}
The first term of the cost function indicates that, in principle, the trader would like to maximise the wealth acquired by following the trading strategy $\theta$. If the investor arrives the terminal time with a non-zero inventory $Q_T^\theta$, the third term of the cost function $Q_T^\theta (P_T^\theta - \alpha Q_T^\theta)$ ensures that these leftovers are executed with a penalisation $\alpha>0$. Finally, the term $- \phi \int_0^T                                                                                               (Q_t^\theta)^2 dt$ penalises holding inventory for a long time. There are different interpretations for this term. For instance, this running inventory penalty could be seen as an \textit{urgency term}. Another interpretation comes from the setting where the investor would like to account for model uncertainty: the larger $\phi$ is, the less certain the trader is about the dynamics imposed on the midprice (see \cite{cartea2, cartea4}). In any case, a large $\phi$ would increase the trading speed near the beginning, and reduce it near the end.

This particular cost function was chosen due to its popularity in the literature \cite{cartea1, eyal, cartea2, transient1, transient2, cartea3, transient3}, but the authors would like to emphasise that the methodology proposed in this paper also applies to other alternative definitions of the cost function, and we are not restricted to this particular choice of $\mathcal C^\theta$.

Properties of signatures, and the shuffle product property \eqref{eq:shuffle product property} in particular, will make finding the optimal trading speed for the optimal control problem \eqref{optimal execution problem} in the restricted space $\mathcal T_{sig}\subset \mathcal T$ easier to solve. Due to the density result stated in Lemma \ref{lemma:density}, we will restrict the space of trading speeds from $\mathcal T$ to $\mathcal T_{sig}$, so that we will solve the following problem instead:

\begin{equation}\label{optimal execution problem sigs}
\sup_{\theta\in \mathcal T_{sig}} \mathbb E[\mathcal C^\theta (\widehat {\mathbb X})].
\end{equation}

\section{Optimal execution}\label{sec:optimal execution}

The cost function \eqref{eq:cost} is a nonlinear function of the underlying price path. However, for signature trading strategies $\theta\in \mathcal T_{sig}$ it turns out to be a linear function on the signature of the midprice process. This is due to the shuffle product property \eqref{eq:shuffle product property} -- each term in the cost function can be rewritten as a linear function on the signature of the midprice process.

\begin{lemma}
Let $\theta\in \mathcal T_{sig}$ be the signature trading speed given $\theta(\widehat{\mathbb X}|_{0, t}) = \langle \ell, \widehat{\mathbb X}_{0, t}^{<\infty}\rangle$, with $\ell \in T((\mathbb R^2)^\ast)$. Then, given any $\widehat{\mathbb X}\in \widehat \Omega_T^p$ and $t\in [0, T]$, we have

\begin{enumerate}
\item $W_t^\ell = \left \langle \left ((\word{2} + \word{\varnothing} - g^\ell)\shuffle \ell \right )\word{1}, \widehat{\mathbb X}_{0,t}^{<\infty}\right \rangle$.
\item $Q_t^\ell = \langle q_0 \word{\varnothing} - \ell \word{1}, \widehat{\mathbb X}_{0, t}^{<\infty}\rangle.$
\item $\int_0^t (Q_s^\ell)^2 ds = \langle (q_0 \word{\varnothing} - \ell \word{1})^{\shuffle 2} \word{1}, \widehat{\mathbb X}_{0, t}^{<\infty}\rangle.$
\item $Q_t^\ell(P_t^\ell - \alpha Q_t^\ell) = \langle (q_0\word{\varnothing} - \ell \word{1})\shuffle (\word{2} + \word{\varnothing} - g^\ell) - \alpha (q_0 \word{\varnothing} - \ell \word{1})^{\shuffle 2}, \widehat{\mathbb X}_{0, t}^{<\infty}\rangle$.
\end{enumerate}

\end{lemma}

\begin{proof}  Let $\widehat{\mathbb X}\in \widehat \Omega_T^p$ and $t\in [0, T]$.

\begin{enumerate}
\item Notice that, because $X_0=1$, we have $X_s = \langle \word{2} + \word{\varnothing},\widehat{\mathbb X}_{0,s}^{<\infty}\rangle $ for each  $s\in [0, t]$. Then, by the shuffle product property \eqref{eq:shuffle product property},

\begin{align*}
W_t^\ell &= \int_0^t P_s^\ell \langle \ell, \widehat{\mathbb X}_{0, s}^{<\infty}\rangle ds = \int_0^t (X_s - \langle g^\ell, \widehat{\mathbb X}_{0,s}^{<\infty}\rangle )\langle \ell, \widehat{\mathbb X}_{0,s}^{<\infty}\rangle ds \\
&= \int_0^t \left \langle (\word{2} + \word{\varnothing} - g^\ell)\shuffle \ell, \widehat{\mathbb X}_{0,s}^{<\infty}\right \rangle ds = \left \langle \left ((\word{2} + \word{\varnothing} - g^\ell)\shuffle \ell \right )\word{1}, \widehat{\mathbb X}_{0,t}^{<\infty}\right \rangle.
\end{align*}
\item Follows from the fact that $\int_0^t \langle \ell, \widehat{\mathbb X}_{0, s}^{<\infty}\rangle ds = \langle \ell \word{1}, \widehat{\mathbb X}_{0, t}^{<\infty}\rangle$.
\item Using (ii),

$$\int_0^t (Q_s^\ell)^2 ds = \int_0^t \langle (q_0 \word{\varnothing} - \ell \word{1})^{\shuffle 2}, \widehat{\mathbb X}_{0, s}^{<\infty}\rangle ds = \langle (q_0 \word{\varnothing} - \ell \word{1})^{\shuffle 2}\word 1, \widehat{\mathbb X}_{0, t}^{<\infty}\rangle.$$

\item Using (ii) again,
\begin{align*}
Q_t^\ell(P_t^\ell - \alpha Q_t^\ell)& = \langle q_0 \word{\varnothing} - \ell \word{1}, \widehat{\mathbb X}_{0,t}^{<\infty}\rangle \langle \word 2 + \word \varnothing -  g^\ell - \alpha (q_0 \word{\varnothing} - \ell \word{1}), \widehat{\mathbb X}_{0,t}^{<\infty}\rangle\\
&=\langle (q_0\word{\varnothing} - \ell \word{1})\shuffle (\word{2} + \word{\varnothing} - g^\ell) - \alpha (q_0 \word{\varnothing} - \ell \word{1})^{\shuffle 2}, \widehat{\mathbb X}_{0, t}^{<\infty}\rangle.
\end{align*}
\end{enumerate}

\end{proof}

Therefore, the optimal liquidation problem \eqref{optimal execution problem} is then transformed into the following problem:

\begin{proposition}\label{prop:cost with signatures}
Let $\theta\in \mathcal T_{sig}$ be the signature trading speed given $\theta(\widehat{\mathbb X}|_{0, t}) = \langle \ell, \widehat{\mathbb X}_{0, t}^{<\infty}\rangle$, with $\ell \in T((\mathbb R^2)^\ast)$. Then, given any $\widehat{\mathbb X}\in \widehat \Omega_T^p$ and $t\in [0, T]$, the cost function can be written as \begin{align*}
\mathcal C^\theta(\widehat{\mathbb X}) &= \left \langle \left ((\word{2} + \word{\varnothing} - g^\ell)\shuffle \ell \right )\word{1} - (q_0 \word{\varnothing} - \ell \word{1})^{\shuffle 2}( \phi\word{1} + \alpha \word \varnothing) + (q_0\word{\varnothing} - \ell \word{1})\shuffle (\word{2} + \word{\varnothing} - g^\ell), \widehat{\mathbb X}_{0, T}^{<\infty}\right \rangle.
\end{align*} Therefore, the optimal liquidation problem \eqref{optimal execution problem} is reduced to

\begin{align}\label{eq:reduced}
\sup_{\ell\in T((\mathbb R^2)^\ast)}  \big \langle \left ((\word{2} + \word{\varnothing} - g^\ell)\shuffle \ell \right )\word{1} - (q_0 \word{\varnothing} - \ell \word{1})^{\shuffle 2}( \phi\word{1} + \alpha \word \varnothing) \\
+(q_0\word{\varnothing} - \ell \word{1})\shuffle (\word{2} + \word{\varnothing} - g^\ell), \mathbb E\left [ \widehat{\mathbb X}_{0, T}^{<\infty}\right ] \big  \rangle.\notag
\end{align}

\end{proposition}
The cost function $\mathcal C^\theta(\widehat{\mathbb X})$ depends on two aspects: a stochastic component and the control $\theta$. Moreover, this dependency is nonlinear. Proposition \ref{prop:cost with signatures} separates this dependency into a deterministic component that solely depends on the control, and on a stochastic component that does not depend on the control. Moreover, because this separation makes the cost function linear on the path, the expectation in \eqref{eq:cost} is moved inside linear functional -- in other words, the resulting optimisation problem \eqref{eq:reduced} depends on the expected signature of the midprice process.

The expected signature of the midprice process is the only dependency on the stochastic process. This object plays the analogous role of the moments of a random variable, but on path space. It was shown in fact in \cite{chevyrev2016characteristic} that under certain growth assumptions, the expected signature determines the law of the stochastic process. Therefore, the fact that \eqref{eq:reduced} depends on the expected signature of the midprice process essentially implies that the optimisation problem depends on the entire law of the process.

\subsection{Numerically solving the optimal execution problem}

\begin{figure}
\centering
\includegraphics[width=0.7\linewidth]{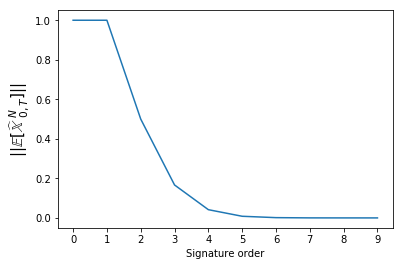}
\caption{$\left \lVert \mathbb E\left [ \widehat{\mathbb X}_{0,T}^{N}\right ]\right \rVert$ as a function of $N$ in the case where the midprice process is a Brownian motion. The factorial decay of the signature makes higher order terms small compared to the first few terms.}
\label{fig:signature decay}
\end{figure}

The optimisation problem \eqref{eq:reduced} from Proposition \ref{prop:cost with signatures} involves the full expected signature $\mathbb E\left [ \widehat{\mathbb X}_{0,T}^{<\infty}\right ]$. In practice, however, one has to consider the truncated expected signature of order $N\in \mathbb N$, i.e. $\mathbb E\left [ \widehat{\mathbb X}_{0,T}^{\leq N}\right ]$.

However, the fast decay of the signature -- it decays factorially -- implies that the first few terms will dominate the rest, and not much information will be lost in the truncation. As a consequence, the expected signature typically also decays factorially (\cite{chevyrev2016characteristic} for instance showed this fact for wide classes of L\'evy, Markov and Gaussian processes).

Figure \ref{fig:signature decay} shows $\left \lVert \mathbb E\left [ \widehat{\mathbb X}_{0,T}^{N}\right ]\right \rVert$ plotted against $N$ in the case where the midprice process $X$ is a Brownian motion. As we see, the factorial decay makes higher order terms small compared to the first few terms. Therefore, in practice one doesn't need to consider truncations of very high order.

Once the signature is truncated at a certain level $N\in \mathbb N$, the optimisation problem \eqref{eq:reduced} consists of finding the global maximum of a certain polynomial in several variables. For example, one can show that if a linear permanent and temporary market impact is considered, the polynomial is a quadratic polynomial, and finding the optimal trading speed will be reduced to finding the (unique) global maximum of a quadratic polynomial in several variables.

Regarding the computation of the truncated expected signature, Monte Carlo methods can be used for this task. Therefore, the only knowledge about the midprice process that is needed to solve the optimal execution problem is how to sample from the path. The signature of a single realisation can be computed using publicly available software such as \texttt{esig}\footnote{\url{https://pypi.org/project/esig/}} or \texttt{iisignature}\footnote{\url{https://pypi.org/project/iisignature/}}. 

\section{Extensions}\label{sec:extensions}

In Section \ref{sec:optimal execution}, we studied a certain optimal liquidation problem. In the present section we analyse different extensions of the problem, and we study how they fit in our framework.

\subsection{Modelling the execution price with exogeneous information}\label{sec:exogeneous information}

For $\theta\in \mathcal T_{sig}$, in Section \ref{subsec:market impact} the market impact was defined as a function of the trading speed and the unaffected midprocess:
\begin{equation}\label{eq:exogeneous information original}
P_t^\theta := X_t - \langle g^\theta, \widehat{\mathbb X}^{<\infty}_{0, t}\rangle,\quad \mbox{with }g^\theta\in T((\mathbb R^2)^\ast).
\end{equation}
However, there are other factors that affect the impact of a trading order \cite{exogeneousimpact1, exogeneousimpact2}. For instance, one may want to incorporate the total traded volume $V:[0, T] \to \mathbb R$ in the market impact \cite{exogeneousimpact2}. Moreover, correlation and cross-asset impact between similar assets will also play a role: the execution price of an order may depend on the midprice process of other assets \cite{exogeneousimpact1, crossasset, crossasset2}.

This feature can be incorporated to our framework, by modelling the execution price by
\begin{equation}\label{eq:exogeneous information}
P_t^\theta := \langle f^\theta, \widehat{\mathbb Z}_{0,t}^{<\infty}\rangle,\quad \mbox{with }f^\theta\in T((\mathbb R^{n+3})^\ast)
\end{equation}
where $\widehat{\mathbb Z}_{0,t}^{<\infty}$ is the signature of $\widehat{Z}_t := (t, X_t, V_t, Y_t^1, \ldots, Y_t^n)\in \mathbb R^{n+3}$, with $V_t$ the total traded volume up to time $t$ and $Y_t^1, \ldots, Y_t^n$ are the midprice processes of $n$ alternative assets that the trader believes that affect the execution price of the main asset. Notice that \eqref{eq:exogeneous information original} is a particular case of \eqref{eq:exogeneous information}. Other exogenous information may also be added to $\widehat Z$.

The methodology proposed in this paper will then apply to this setting: the optimisation problem \eqref{optimal execution problem}, for the new definition of market impact, will be reduced to an optimisation problem similar  to \eqref{eq:reduced}, namely
\begin{equation}\label{eq:optimisation exogeneous information}
\sup_{\ell\in T((\mathbb R^2)^\ast)}  \big \langle \left (f^\ell\shuffle \ell \right )\word{1} - (q_0 \word{\varnothing} - \ell \word{1})^{\shuffle 2}( \phi\word{1} + \alpha \word \varnothing) +(q_0\word{\varnothing} - \ell \word{1})\shuffle f^\ell, \mathbb E\left [ \widehat{\mathbb Z}_{0, T}^{<\infty}\right ] \big  \rangle.
\end{equation}

\subsection{Optimal trading, as opposed to liquidation}

In this paper we have been focusing on the case where a trader has an initial inventory at $t=0$, and she would like to get rid of it by time $t=T$. However, certain high-frequency traders may be interested in the following alternative question: if one starts with no inventory at $t=0$ and one would like to finish with no inventory at $t=T$, what is the best trading strategy that can be followed on $[0, T]$? This paper's framework can be modified for this purpose by redefining the inventory $Q_t$ in Definition \ref{def:wealth inventory cost} by setting $q_0=0$.

\subsection{Cross-asset portfolio liquidation}

The discussion on Section \ref{sec:exogeneous information} suggests another extension of the original problem studied in this paper. Suppose there are $n$ assets $Y^1, \ldots, Y^n$ and a trader has an initial portfolio $q=(q_1, \ldots, q_n)\in \mathbb R_+^n$. If the trader wishes to liquidate the inventory $q$ (see \cite{crossasset, crossasset2}), she can consider an optimal control problem similar to \eqref{eq:optimisation exogeneous information} that incorporates her risk profile.

More generally, the trader could aim to transition from a starting portfolio $q_{start}\in \mathbb R^n$ on $n$ traded assets, to a final portfolio $q_{end}\in \mathbb R^n$, and she would like to do so in an optimal way. Again, our framework can be adapted for this task.

%

\subsection{Other cost functions}

The cost function considered in \eqref{optimal execution problem} was chosen in order to be consistent with the literature \cite{cartea1, eyal, cartea2, transient1, transient2, cartea3, transient3}. However, the methodology we propose is not intrinsic to this cost function, and it can be applied to other cost functions that the trader may find more appropriate.

\section{Numerical experiments}\label{sec:numerical experiments}

In this section we implement the proposed methodology and test it on different settings. We begin by showing that, when we apply the methodology to various settings studied in the literature, we retrieve the existing results, thus reaffirming that our framework is a generalisation of many frameworks considered in the literature and validating the trading strategy returned by the signature methodology. Then, we apply our approach to new settings.

\subsection{Brownian motion with temporary and permanent market impact}\label{subsec:numerical bm}

\begin{figure}
\centering
\includegraphics[width=0.7\linewidth]{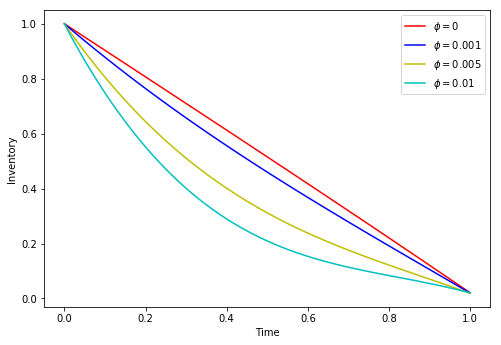}
\caption{The trader's inventory for 100 midprice path realisations and the setting considered in Section \ref{subsec:numerical bm}. Different running inventory penalties $\phi$ were considered.}
\label{fig:numerical bm}
\end{figure}

In this section we will consider the framework studied in \cite[Section 6.5]{carteabook}. We will assume that the unaffected midprice process follows a Brownian motion with volatility $\sigma$, that is, $X_t:= \sigma W_t$ with $\sigma>0$ and $W$ a Brownian motion. For a signature trading speed $\theta\in \mathcal T_{sig}$ given by $\theta(\widehat{\mathbb X}|_{[0,t]})=\langle \ell, \widehat{\mathbb X}_{0,t}^{<\infty}\rangle$ with $\ell\in T((\mathbb R^2)^\ast)$, the execution price will be given by a \textit{permanent} market impact and a \textit{temporary} market impact:
$$P_t^\theta := X_t - k\int_0^t \theta(\widehat{\mathbb X}|_{[0, s]})ds - \lambda \theta(\widehat{\mathbb X}|_{[0,t]}),$$ with $k,\lambda>0$.

It was mentioned in Section \ref{subsec:market impact} that this market impact is included in our framework. More specifically, we have:
$$P_t^\ell = X_t - k \int_0^t \langle \ell, \widehat{\mathbb X}_{0,s}^{<\infty}\rangle ds - \lambda \langle \ell , \widehat{\mathbb X}_{0,t}^{<\infty}\rangle = X_t - \langle g^\ell, \widehat{\mathbb X}_{0,t}^{<\infty}\rangle$$ with $g^\ell := k \ell \word 1 + \lambda \ell$.

We may then solve \eqref{eq:reduced}. The chosen parameters were $q_0=1$, $\lambda=10^{-3}$, $k=10^{-4}$, $\alpha=10$, $\sigma=0.02$ and $T=1$, and different values for $\phi$ were considered. Truncated signatures of order 7 were considered to solve \eqref{eq:reduced}. As it has been established in the literature (see \cite[Section 6.5]{carteabook}) the optimal trading speed does not depend on the midprice. Moreover, if we set $\phi=0$ so that no running inventory penalties are considered, it is known that the optimal trading speed is constant. On the other hand, when $\phi$ is increased, the trader decides to liquidate the inventory sooner. All this features are captured in the results we obtained -- see Figure \ref{fig:numerical bm}.

\subsection{Incorporating trading signals}\label{subsec:signals}

\begin{figure}
\centering
\includegraphics[width=0.7\linewidth]{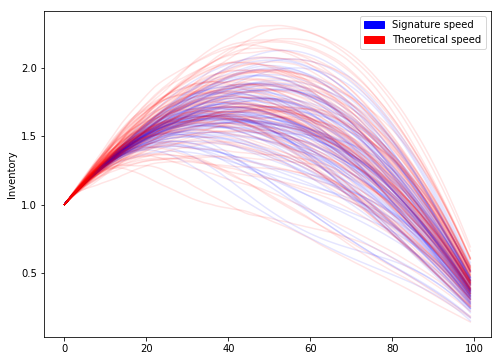}
\caption{The trader's inventory for 100 midprice path realisations and the setting considered in Section \ref{subsec:signals}, both for the theoretical optimal speed (red) and the signature trading speed (blue).}
\label{fig:numerical signal}
\end{figure}

Lehalle and Neuman in \cite{eyal} considered an optimal liquidation problem where the investor has access to some trading signal that predict short-term price movements, such as order book imbalance.

In this case, the midprice process was taken to be $X_t:=\int_0^t I_sds + \sigma W_t$, where $I$ is the signal process, $\sigma>0$ is volatility and $W$ is a Brownian motion. In the original paper \cite{eyal}, the signal $I$ that was considered was an Ornstein-Uhlenbeck process $dI_t = -\gamma I_tdt + \sigma_0dW_t$, where $\gamma,\sigma>0$ are constants. Therefore, given that the midprice process is a semimartingale, this example also falls within our framework.

The price impact that was considered in \cite{eyal} was a linear temporary price impact. Therefore, the execution price will be given by \eqref{eq:execution price}, where $g^\ell := \lambda \ell$ with $\lambda>0$.

Figure \ref{fig:numerical signal} shows the running inventory for 100 realisations of the midprice process, both for the signature trading speed and the optimal trading speed that was derived in \cite{eyal}. The chosen parameters were $q_0=1$, $\lambda=10^{-3}$, $\alpha=10^{-2}, \phi=10^{-3}, I_0=0.02$ and $\gamma=0.1$. Truncated signatures of order 9 were considered. As we see, the signature trading speed seems to be a close approximation of the theoretical optimal speed. The numerical expected cost of the signature trading speed is 1.0169981 whereas the optimal trading speed's expected cost is 1.0170877.

Notice that the presence of the signal in the midprice process introduces a positive drift, and therefore, as illustrated by Figure \ref{fig:numerical signal}, it is optimal to begin by purchasing shares in order to sell them for a profit later. This could be avoided by increasing the running inventory penalty $\phi$.

\subsection{Incorporating order-flow}\label{subsec:order flow}

\begin{figure}
\centering
\includegraphics[width=0.7\linewidth]{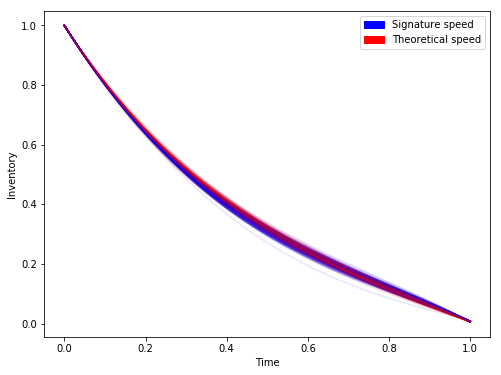}
\caption{The trader's inventory for 100 midprice path realisations and the setting considered in Section \ref{subsec:order flow}, both for the theoretical optimal speed (red) and the signature trading speed (blue).}
\label{fig:order flow}
\end{figure}

In \cite{cartea2}, the authors incorporate the order-flow of all agents into the midprice dynamics. This is done by considering the midprice process
$$X_t := k\int_0^t (\mu_s^+ - \mu_s^-)ds + \sigma W_s,$$
where $\mu_t^+$ and $\mu_t^-$ are the aggregated buying and selling orders of all market participants, respectively. These orders are assumed to follow the dynamics
$$d\mu_t^\pm =-\kappa\mu_t^\pm dt + \eta_{1+L_{t-}^\pm}^\pm dL_t^\pm$$
with $L_t^\pm$ independent Poisson processes of intensity $\lambda_0$, and $\eta_i^\pm \sim Exp(\eta_0 \kappa)$ has an exponential distribution. Moreover, a temporary market impact $\lambda \theta(\widehat{\mathbb X})$ was included as well.

Figure \ref{fig:order flow} shows the inventory for 100 realisations of the midprice path, both for the signature trading speed and the optimal trading speed that was derived in \cite{cartea2}. The expected cost function of the signature trading speed is 0.995690, very close to the expected cost of the optimal speed: 0.995722. The parameters we considered are $\lambda=5\cdot 10^{-4}$, $k=10^{-4}$, $q_0=1$, $\alpha=2$, $\phi=5\cdot 10^{-3}$, $\sigma=0.1$, $\kappa = \lambda_0=5$, $\eta_0=0.8$ and signatures of order 7.

\subsection{Fractional Brownian motion}\label{subsec:fractional bm}

\begin{figure}
\centering
\includegraphics[width=0.48\linewidth]{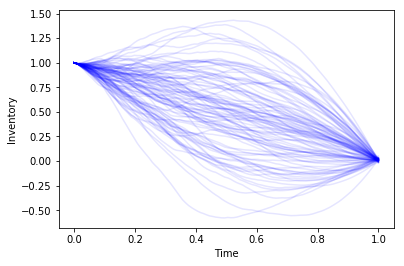}
\includegraphics[width=0.48\linewidth]{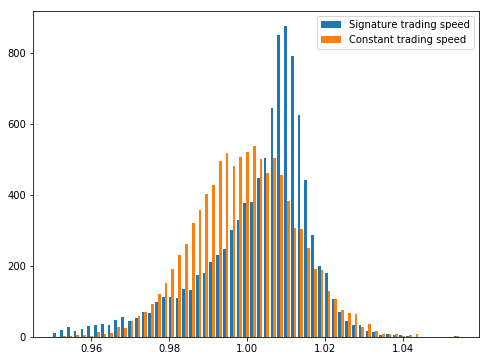}
\caption{Trader's inventory (left) and the trader's wealth distribution (right) in the case where the midprice process is a fractional Brownian motion.}
\label{fig:fbm}
\end{figure} 

In this section, we assume that the midprice process $X_t$ is a fractional Brownian motion. We assume a linear market impact. In other words, the execution price will be given by
$$P_t^\theta := \sigma W_t^H - \lambda \theta(\widehat{\mathbb X}|_{[0,t]}),$$
where $W_t^H$ is a fractional Brownian motion with Hurst parameter $H$, and $\sigma, \lambda>0$ are constants.

Figure \ref{fig:fbm} shows the midprice and inventory in the case where $H=1/3$, $\sigma=0.02$, $q_0=1$, $\phi=0$, $\lambda=10^{-3}$, $\alpha=0.1$, $T=1$ and truncated signatures of order 7 are considered.

As we see, the behaviour differs significantly from the case where $H=1/2$ (i.e. when $X_t$ is a Brownian motion). Indeed, given that we don't include a running inventory penalty as $\phi=0$, in the Brownian case we would expect the inventory $Q_t$ to be linear. However, Figure \ref{fig:fbm} illustrates that this is not the case for the fractional Brownian motion, and the trading speed depends strongly on the midprice process. In fact, if we look at the expected cost of the constant trading speed -- it is given by 0.9991335 -- we see that the signature trading speed for the fractional Brownian motion outperforms the constant trading speed strategy -- the expected cost of the signature trading speed is 1.0031300. This is outperformance of the signature trading speed is reflected in the wealth distribution of both strategies shown in Figure \ref{fig:fbm} (right).

\section{Experiments with market data}\label{sec:market data}

To solve \eqref{eq:reduced}, the only information that is needed about the midprice process is its expected signature. In this section, we use real market data to estimate the expected signature, which is then used to solve \eqref{eq:reduced}. Then, we evaluated the performance of the optimal execution strategy in an out-of-sample set of market paths.

We considered midprice market data of Apple (AAPL) for 1 year, from the \nth{1} of January 2018 to the \nth{31} of December 2018, which was obtained from LOBSTER\footnote{\url{https://lobsterdata.com/}}. This data was divided into a training set of 10 months (January--October) and an out-of-sample set of 2 months (November--December).

We considered 15 minute windows from different times of each trading day -- more specifically, we considered 10:00--10:15, 11:00--11:15, 12:00--12:15 and 13:00--13:15. We estimated the expected signature over each of these 15 minute windows by computing the empirical expectation of the signature (signatures of order 13 were considered) of the corresponding 15 minute windows from the testing set. Therefore, to some extent, we assume that the midprice process follows a similar behaviour over each of the windows throughout the trading year.

\begin{figure}
  \begin{subfigure}[b]{0.5\linewidth}
    \centering
    \includegraphics[width=0.75\linewidth]{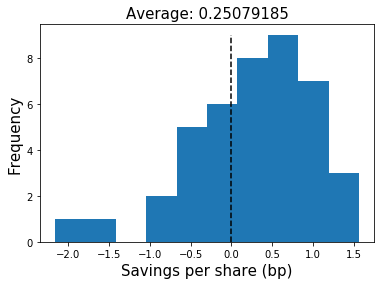} 
    \caption{10:00--10:15.} 
    \label{fig7:a} 
    \vspace{4ex}
  \end{subfigure}
  \begin{subfigure}[b]{0.5\linewidth}
    \centering
    \includegraphics[width=0.75\linewidth]{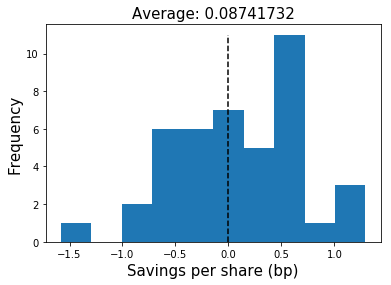} 
    \caption{11:00--11:15.} 
    \label{fig7:b} 
    \vspace{4ex}
  \end{subfigure} 
  \begin{subfigure}[b]{0.5\linewidth}
    \centering
    \includegraphics[width=0.75\linewidth]{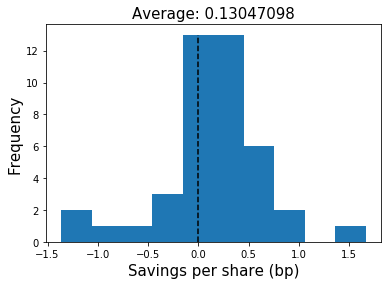} 
    \caption{12:00--12:15.} 
    \label{fig7:c} 
  \end{subfigure}
  \begin{subfigure}[b]{0.5\linewidth}
    \centering
    \includegraphics[width=0.75\linewidth]{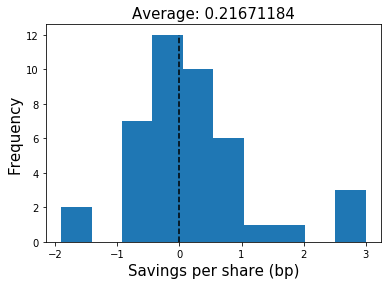} 
    \caption{13:00--13:15.} 
    \label{fig7:d} 
  \end{subfigure} 
  \caption{Out-of-sample performance of the signature approach to optimal liquidation, compared to the Almgren--Chriss benchmark. The optimal signature trading speed consistently outperforms the benchmark across all 15-minute windows.}
  \label{fig:market data} 
\end{figure}

Once the expected signature of the midprice process for each of the 15 minute windows is estimated from the training set, we solved the optimisation problem \eqref{eq:reduced} to estimate the optimal signature trading speed. We included a temporary and market impact:
$$P_t^\theta := X_t - k\int_0^t \theta(\widehat{\mathbb X}|_{[0, s]})ds - \lambda \theta(\widehat{\mathbb X}|_{[0,t]}).$$
The parameters we used where $\lambda =10^{-3}$, $k=10^{-4}$, $\alpha=0.1, \phi=10^{-4}$ and $q_0=1$.
 We then evaluated the performance on the out-of-sample set for each of the 15 minute windows. Following \cite{cartea2}, we compared the performance against the Almgren--Chriss execution strategy \cite{almgren}. More specifically, we considered the \textit{savings per share} metric (in basis points) that was used in \cite{cartea2}, which is defined by
$$\dfrac{W_T - W_T^{AC}}{W_T^{AC}} \times 10^4,$$
where $W_T$ and $W_T^{AC}$ are the terminal wealth of the optimal signature trading speed and Almgren--Chris execution strategy, respectively.

The results are shown in Figure \ref{fig:market data}. The optimal signature trading speed outperforms the Almgren--Chriss benchmark on all 15-minute windows, as on average the savings per share of the signature trading speed is positive.

Notice that the only assumption we have made is that the midprice process behave similarly on the same 15-minute window across different trading days. Other than that, our approach is model-free: we can, in a nonparametric and model-free way, estimate the optimal trading speed from market data.

\section{Conclusion}

In this paper we propose a methodology to numerically approximate the solution of certain optimal execution problems. This is done in the general framework of geometric rough paths, which in particular contains many existing models in the literature.

Rough path signatures provide a methodology to reduce the original optimisation problem into a finite-dimensional, computationally feasible optimisation problem. The only information that is needed from the underlying price process is its expected signature, which can be computed using Monte Carlo methods.

This approach was tested in Section \ref{sec:numerical experiments}, where we show that in those cases where the optimal trading speed is known, the signature-based numerical approach is capable of retrieving it. Moreover, the generality of the approach allows the estimation of the optimal trading speed in those settings where the optimal solution is unknown. In Section \ref{sec:market data} on the other hand, we showed how our methodology can be used in real market data and we demonstrated that the signature approach outperforms the Almgren--Chriss benchmark.

\appendix
\section{Proofs} \label{appendix:proofs}

\begin{proof}[Proof of Lemma \ref{lemma:density}]\label{proof:lemma density}
Let $\varepsilon>0$. Because $\widehat \Omega_T^p$ is reflexive, there exists $\mathcal K\subset \widehat\Omega_T^p$ compact such that $\mathbb P[\mathcal K]>1-\varepsilon$.

Let $\theta_1,\theta_2\in \mathcal T_{sig}$. Then, by definition there exist linear functionals $\ell_1,\ell_2\in T((\mathbb R^2)^\ast)$ such that $\theta_i(\widehat{\mathbb X}|_{[0,t]})=\langle \ell_i,\widehat{\mathbb X}_{0,t}^{<\infty}\rangle$ for all $\widehat{\mathbb X}|_{[0,t]}\in \Lambda_T, i=1,2$. Define \linebreak $\theta(\widehat{\mathbb X}|_{[0,t]}):=\langle \ell_1 \shuffle \ell_2, \widehat{\mathbb X}_{0,t}^{<\infty}\rangle$. Then, by the shuffle product property \eqref{eq:shuffle product property} we have
\begin{align*}
\theta_1(\widehat{\mathbb X}|_{[0,t]})\theta_2(\widehat{\mathbb X}|_{[0,t]}) &=\langle \ell_1, \widehat{\mathbb X}_{0,t}^{<\infty}\rangle\langle \ell_2, \widehat{\mathbb X}_{0,t}^{<\infty}\rangle\\
&=\langle \ell_1\shuffle \ell_2, \widehat{\mathbb X}_{0,t}^{<\infty}\rangle\\
&=\theta(\widehat{\mathbb X}|_{[0,t]}).
\end{align*}
Therefore, and because the sum of two signature trading speeds is trivially  a signature trading speed, $\mathcal T_{sig}$ form an algebra. On the other hand, the uniqueness of the signature (Corollary \ref{cor:uniqueness}) implies that $\mathcal T_{sig}$ separates points. Indeed, given $\widehat{\mathbb X}|_{[0,t]}, \widehat{\mathbb Y}|_{[0,t]}\in \widehat \Omega_T^p$ distinct, because we have $\widehat{\mathbb X}_{0,t}^{<\infty}\neq \widehat{\mathbb Y}_{0,t}^{<\infty}$ we immediately have that there exists $\ell\in T((\mathbb R^2)^\ast)$ such that $\langle \ell, \widehat{\mathbb X}_{0,t}^{<\infty}\rangle \neq \langle \ell, \widehat{\mathbb Y}_{0,t}^{<\infty}\rangle$. Moreover, $\mathcal T_{sig}$ contains constants, as $\langle \word \varnothing, \widehat{\mathbb X}_{0,t}^{<\infty}\rangle=1$ for all $\widehat{\mathbb X}|_{[0,t]}\in \widehat \Omega_T^p$. Therefore, by Stone--Weierstrass theorem we conclude that $\mathcal T_{sig}$, restricted to $\mathcal K$, is dense in $\mathcal T$.
\end{proof}

\section*{Acknowledgments}
This work was supported by The Alan Turing Institute under the EPSRC grant EP/N510129/1. The authors would like to thank Sebastian Jaimungal, Alvaro Cartea and Leandro Leandro Sanchez Betancourt for reading a preprint of this paper and for giving their insights.

\section*{Disclosure statement}

Opinions and estimates constitute our judgement as of the date of this Material, are for informational purposes only and are subject to change without notice. This Material is not the product of J.P. Morgans Research Department and therefore, has not been prepared in accordance with legal requirements to promote the independence of research, including but not limited to, the prohibition on the dealing ahead of the dissemination of
investment research. This Material is not intended as research, a recommendation, advice, offer or solicitation for the purchase or sale of any financial product or service, or to be used in any way for evaluating the merits of participating in any transaction. It is not a research report and is not intended as such. Past performance is not indicative of future results. Please consult your own advisors regarding legal, tax, accounting or any other aspects including suitability implications for your particular circumstances. J.P. Morgan disclaims any responsibility or liability whatsoever for the quality, accuracy or completeness of the information herein, and for any reliance on, or use of this material in any way.

Important disclosures at: \url{www.jpmorgan.com/disclosures}.

\bibliographystyle{alpha}
\bibliography{references}

\end{document}